\documentclass[conference,10pt]{IEEEtran}
\usepackage[utf8]{inputenc}
\usepackage{graphicx}
\usepackage{amsmath}
\usepackage{amssymb}
\usepackage{listings}
\usepackage{xcolor}
\usepackage{booktabs}
\usepackage{multirow}
\usepackage{tikz}
\usepackage{balance}
\usepackage{seqsplit}
\usepackage{caption}
\usepackage{xspace}
\usepackage{hyperref}
\usetikzlibrary{shapes.geometric, arrows, positioning, fit, backgrounds}

\hypersetup{
    colorlinks=true,
    linkcolor=black,
    citecolor=black,
    urlcolor=black
}

\lstdefinelanguage{json}{
    basicstyle=\ttfamily\footnotesize,
    morestring=[b]",
    stringstyle=\color{red!60!black},
    literate=
     *{:}{{{\color{blue}:}}}{1}
      {,}{{{\color{blue},}}}{1}
      {\{}{{{\color{blue}\{}}}{1}
      {\}}{{{\color{blue}\}}}}{1}
      {[}{{{\color{blue}[}}}{1}
      {]}{{{\color{blue}]}}}{1}
}

\lstdefinelanguage{Solidity}{
    keywords={function, external, internal, public, private, view, pure, payable, returns, return, if, else, for, while, memory, calldata, storage, string, uint256, address, bytes, bytes32, bool, mapping, struct, contract, import, pragma, require, emit, event, modifier},
    morecomment=[l]{//},
    morecomment=[s]{/*}{*/},
    morestring=[b]",
    morestring=[b]',
    sensitive=true
}

\lstdefinelanguage{Rust}{
    keywords={fn, let, mut, if, else, for, while, loop, match, return, pub, struct, enum, impl, use, mod, self, Self, Ok, Err, Some, None, Result, Option, String, Vec, u128, u64, u32, u8, bool},
    morecomment=[l]{//},
    morecomment=[s]{/*}{*/},
    morestring=[b]",
    sensitive=true
}

\newtheorem{theorem}{Theorem}

\newcommand{\eg}{\textit{e.g.}\xspace}
\newcommand{\etal}{\textit{et al.}}

\begin{document}

\title{An EVM--Cosmos Bridge for Atomic, Semantic-Preserving Transfer of Environmental Tokens}

\author{
    \IEEEauthorblockN{
        Sebastian Pandolfi\IEEEauthorrefmark{1},
        Guzm\'an Llamb\'ias\IEEEauthorrefmark{1}\IEEEauthorrefmark{2},
        Laura Gonz\'alez\IEEEauthorrefmark{1}
    }
    \\
    \IEEEauthorblockA{
        \IEEEauthorrefmark{1}Facultad de Ingenier\'ia, Universidad de la Rep\'ublica, Montevideo, Uruguay\\
        \{sebastian.pandolfi, gllambi, lauragon\}@fing.edu.uy
    }
    \IEEEauthorblockA{
        \IEEEauthorrefmark{2}Pyxis Research, Pyxis, Montevideo, Uruguay\\
        guzman.llambias@pyxis.tech
    }
}

\maketitle

\begin{abstract}
Transferring semantically typed tokens between heterogeneous blockchains---EVM and Cosmos ecosystems---raises two problems: \textit{semantic mixing}, the loss of a token's domain-specific meaning in transit (\eg, CO$_2$ removals issued under different methodologies), and \textit{atomicity} despite divergent consensus and timing assumptions. We propose a bridge that addresses both. A \textit{deterministic indicator identity} derived from cryptographic commitments, mirrored in class registries on both chains and checked at mint time, preventing semantic mixing by construction. For atomicity, we extend hash time-locked contracts (HTLCs) from two-party exchange to single-owner transfer with an \textit{automatic callback} that burns the source escrow after a destination claim, and a protocol-minimum \textit{bounty} that makes the fallback burn incentive-compatible even against adversarial senders. We machine-check the protocol's invariants with TLC---supply conservation fails without the rational-monitor assumption and holds with it---and validate a bidirectional prototype on Polygon and Neutron mainnets over Axelar GMP: fully automatic transfers with a median end-to-end latency of 97~s over 50 Polygon$\rightarrow$Neutron transfers across 7 days (all completed by the on-chain automatic callback) and 93~s in the reverse direction, and semantic-validation overhead below 0.4\% of messaging cost. We also document the integration failure modes of Cosmos-originated messaging, whose dominant operational risk proves to be relay \emph{eligibility} rather than latency.
\end{abstract}


\section{Introduction}
\label{sec:introduction}

The proliferation of blockchain networks has created a fragmented ecosystem in which assets and data remain siloed across distinct consensus mechanisms, execution environments, and communication protocols~\cite{belchior2021,wang2021}. Cross-chain interoperability has therefore emerged as critical infrastructure, while bridges have become prime attack targets: Zhang~\etal\ document 35 attacks and ten vulnerability classes~\cite{zhang2024bridges}, and Chainalysis estimates \$2B stolen across 13 bridge hacks in 2022 alone~\cite{chainalysisBridges2022}.

A particularly challenging case arises with \textit{domain-specific tokens}---tokens whose units carry application-level meaning rather than being fully fungible. Pandolfi~\etal~\cite{pandolfi2025rapid} survey this challenge for \textit{environmental indicator tokens}, the focus of this paper. These tokens represent quantitative measurements of environmental impact such as CO$_2$ removals, renewable energy generation, or treated-water volumes~\cite{frontiers2024klimadao}. Unlike fully fungible assets, they have \textit{meaning}: 1~kgCO$_2$e of carbon removal issued under Verra~VM0007 is not interchangeable with one issued under Gold~Standard~GS-AFOLU, let alone with 1~kWh of energy. We call such tokens \textit{multi-class tokens}: grouped into discrete classes that are fungible within a class and non-interchangeable across classes. Transferring them across architecturally dissimilar chains poses two problems.

\textit{Problem~1 (Atomicity across heterogeneous chains).} Domain-independent: any cross-chain transfer must avoid leaving tokens on both chains or on neither. Classical two-phase commit~\cite{gray1978} requires a trusted coordinator; HTLCs~\cite{herlihy2018} provide a trustless alternative but suffer from a \textit{timeout race condition} when relay latency stretches between heterogeneous ecosystems~\cite{thyagarajan2022}.

\textit{Problem~2 (Semantic mixing).} A bridge treating all tokens as fungible merges distinct classes into a single balance: Verra- and Gold-Standard-issued CO$_2$ bridged through a standard ERC-20 bridge arrive indistinguishable, with no way to recover methodology, registry, or vintage. The mechanism is documented at the application layer---Toucan's Base~Carbon~Tonne (BCT) pool aggregated heterogeneous credits into one fungible ERC-20, creating a ``lemons'' market in which low-quality ``zombie'' credits displaced higher-integrity ones~\cite{carbonplan2022}. The cross-chain instance has not yet materialised at scale, but the BCT episode confirms the underlying dynamics; the bridge is the natural place to forestall it by design.

To our knowledge, no existing cross-chain approach simultaneously addresses both problems for multi-class tokens across heterogeneous EVM--Cosmos environments~\cite{zamyatin2021,belchior2021,satp2023,wormhole2024,layerzero2021,ccip2023,axelar2024its}.

 We propose a bridge that combines relay-based messaging with an HTLC atomic-transfer protocol augmented by an \textit{automatic callback mechanism}. After the destination claim reveals the secret, the destination contract emits a relayed message that burns the escrowed source tokens, closing the timeout-race window in a bounded delay rather than at $T_e$; a bounty escrowed at lock time and paid to whoever finalises the source-side burn makes the fallback path incentive-compatible. While HTLCs have traditionally been used for atomic \textit{exchange} (two-party swap), we adapt them to atomic \textit{transfer} (single-owner burn-and-mint) and explicitly handle heterogeneous-chain timing. Semantic mixing is prevented by deriving a \textit{deterministic indicator identity} from cryptographic commitments and mirroring a class registry on both chains, restricting destination mints to pre-registered classes. The five contributions, used as a comparison framework with related work (Table~\ref{tab:bridges}), are: \textit{HTLC-based atomicity for cross-chain transfer~(F1)}, \textit{automatic callback~(F2)}, \textit{deterministic indicator identity~(F3)}, \textit{cross-chain semantic-consistency enforcement~(F4)}, and \textit{multi-class transfer with mirrored registry~(F5)}.

We instantiate the design as a bidirectional prototype using ERC-1155 tokens on Polygon and CosmWasm contracts on Neutron, with Axelar GMP as the relay. Mainnet evaluation demonstrates fully automatic transfers in both directions (108~s and 93~s in representative runs), characterizes latency over 50 transfers across 7 days, exercises the refund and monitor-recovery paths, and reports the integration failure modes encountered during deployment; semantic-validation overhead stays below 0.4\% of cross-chain messaging cost. Sections~\ref{sec:background}--\ref{sec:conclusion} present background, related work, design, prototype, evaluation, discussion, and conclusions.

\section{Background}
\label{sec:background}

\subsection{Cross-Chain Transactions, Relayers, and Bridges}
\label{sec:bg-taxonomy}

A \textit{cross-chain transaction} spans independent ledgers; because no ledger natively observes another, an external \textit{relayer} forwards messages, and the surrounding protocol is a \textit{bridge}. Zamyatin~\etal~\cite{zamyatin2021} systematise bridges as \textit{notary}, \textit{relay/sidechain}, and \textit{HTLC} categories---not mutually exclusive; ours combines a relay scheme with HTLCs---and define the \textit{atomicity}, \textit{consistency}, and \textit{liveness} properties we use throughout. Transfers reconcile supply via \textit{burn-and-mint} (irreversible, supply-conserving; our choice) or custodial \textit{lock-and-mint}~\cite{zamyatin2021,robinson2022}.

Axelar~\cite{axelar} exemplifies a relay-scheme bridge: a proof-of-stake validator network verifies and forwards messages, exposing General Message Passing (GMP) on EVM and delivering to Cosmos destinations over IBC~\cite{ibc}. The EVM--Cosmos pair is heterogeneous in three ways that shape the protocol: (i)~\emph{serialisation}---EVM consumes ABI-encoded bytes, CosmWasm JSON, forcing an explicit encoding boundary (Section~\ref{sec:implementation}); (ii)~\emph{finality}---probabilistic on the EVM side vs.\ instant BFT, forcing an asymmetric confirmation policy visible in the leg latencies (Section~\ref{sec:evaluation}); and (iii)~\emph{relayer infrastructure}---EVM$\rightarrow$Cosmos legs are delivered end-to-end by Axelar's own relayer, whereas Cosmos-originated messages depend on public IBC relayers plus fee-incentivised executors, and must carry an explicit relayer fee in an Axelar-registered asset or they are routed but never executed (Section~\ref{sec:eval-failures}).

\subsection{Hash Time-Locked Contracts and the Timeout Race Condition}
\label{sec:bg-htlc}

A hash time-locked contract (HTLC)~\cite{herlihy2018} releases escrowed assets either on timely revelation of a preimage of a cryptographic hash or on expiration of a time bound. A sender escrows under hashlock $H = \text{hash}(S)$ with deadline $T$; the recipient claims by revealing $S$ before $T$, else the sender refunds after $T$. HTLCs underpin the Lightning Network and were formalised by Herlihy~\cite{herlihy2018} for atomic cross-chain swaps. Across two chains, an HTLC uses paired deadlines $T_e$ (source refund) and $T_c$ (destination claim) with $T_c < T_e$; safety relies on the source chain learning of a destination claim before $T_e$.

Across heterogeneous chains this assumption is fragile: if the destination claim lands near $T_c$ and secret propagation exceeds $(T_e - T_c)$, a malicious user can claim on the destination and still refund on the source---the \textit{timeout race condition}, documented in security surveys~\cite{zhang2024bridges,li2025bridges} and addressed cryptographically only for swap protocols~\cite{thyagarajan2022}. Two-phase commit~\cite{gray1978} removes the race via a trusted coordinator, the single point of trust bridges seek to avoid.

\subsection{Environmental Indicator Tokens}
\label{sec:indicators}

Environmental indicator tokens represent quantitative environmental measurements---tonnes of CO$_2$ removed, kilowatt-hours generated, litres of water treated~\cite{pandolfi2025rapid,frontiers2024klimadao}. Each class carries metadata (indicator type, canonical unit, issuance methodology, and cryptographic commitments to a profile document and evidence bundle), and fungibility holds only within a class: methodologies differ in additionality, monitoring, and verification stringency~\cite{frontiers2024klimadao,pandolfi2025rapid}, and differing units are non-interchangeable a fortiori. The BCT episode of Section~\ref{sec:introduction} shows the consequence of collapsing this structure---and Toucan's reactive countermeasures (gating, redemption fees) show why prevention-by-design beats ex-post remediation~\cite{carbonplan2022}. The ERC-1155 multi-token standard~\cite{erc1155} represents non-interchangeable classes within one contract via distinct \texttt{tokenId} values; our design uses ERC-1155 on the source and a mirrored registry on the destination (Section~\ref{sec:design}).

\subsection{System Model and Assumptions}
\label{sec:threat-model}

The protocol's safety claims rest on six assumptions. \textbf{A1 (cryptographic strength):} keccak256 is collision- and second-preimage-resistant on 32-byte inputs. \textbf{A2 (consensus safety):} each chain produces a single canonical history once finality is reached; the source-chain finality is enforced by the relay's confirmation policy (Section~\ref{sec:evaluation}). \textbf{A3 (forward-relay liveness):} the relay delivers the prepare message to the destination at least once within $T_e$; A2+A3 ensure that an initiated transfer can always be either claimed or refunded. \textbf{A4 (rational-monitor callback safety):} if the callback fails, safety requires at least one party with knowledge of $S$ to invoke \texttt{claimBurn} on the source chain before $T_e$ (Section~\ref{sec:disc-callback-gap}); the capable set is unbounded in principle, and the protocol escrows a bounty of at least a protocol-enforced minimum $b_{\min}$---sized above the monitor's execution cost---paid to whichever party does so, reducing the assumption to the existence of one rational, live observer of the destination chain. \textbf{A5 (honest profile distributor at registration):} the issuer publishes a profile matching the on-chain \texttt{profileHash} at registration time. \textbf{A6 (relay message authenticity):} an inbound message accepted by either chain's gateway was emitted by the counterpart bridge contract with the stated payload; in the prototype this is discharged by Axelar's proof-of-stake validator quorum. A relay violating A6 can mint unbacked tokens into \emph{registered} classes but cannot mix classes and cannot burn escrow before a genuine claim publishes $S$ (the callback handler verifies the preimage). The adversary may control users, message delivery (subject to A3), and observers; it may not break A1, A2, A5, or A6. We do not assume an honest user, an honest issuer past registration, or a live relay; message authenticity is the only property for which the relay is trusted. Table~\ref{tab:assumptions} summarises which assumptions each safety property depends on; all property statements in Sections~\ref{sec:state-machine} and~\ref{sec:eval-security} reference it.

\begin{table}[t]
\centering
\caption{Assumption dependencies of the safety properties. ``cond.'' marks properties that depend on A4 only in the callback-failure regime (Section~\ref{sec:disc-callback-gap}).}
\label{tab:assumptions}
\begin{tabular}{@{}lcccccc@{}}
\toprule
\textbf{Property} & A1 & A2 & A3 & A4 & A5 & A6 \\
\midrule
P1 atomicity (per transfer) & \checkmark & \checkmark & -- & cond. & -- & -- \\
P2 supply conservation & \checkmark & \checkmark & \checkmark & cond. & -- & \checkmark \\
P3 semantic pres.\ (class identity) & \checkmark & -- & -- & -- & \checkmark & -- \\
Thm.~\ref{thm:semantic} (identity $+$ amount backing) & \checkmark & -- & -- & -- & \checkmark & \checkmark \\
\bottomrule
\end{tabular}
\end{table}

\section{Related Work}
\label{sec:related}

We compare related work against the five features F1--F5 of Section~\ref{sec:introduction}, organised into four groups.

\subsection{Cross-Chain Protocol Formalisations}

Herlihy~\cite{herlihy2018} formalises atomic cross-chain swaps for fungible exchanges; we adopt his HTLC foundation~(F1) but address its two main limitations under heterogeneous chains: the timeout race (via~F2) and the implicit assumption of semantic equivalence (via~F3, F4). Thyagarajan~\etal~\cite{thyagarajan2022} harden swaps with adaptor signatures but still assume fungible assets. Our protocol is a relay+HTLC \textit{hybrid} in Zamyatin~\etal's taxonomy~\cite{zamyatin2021}. The major surveys~\cite{belchior2021,robinson2022,zhang2024bridges,li2025bridges} identify semantic preservation in no surveyed system; Guo~\etal~\cite{guo2024} combine relaying with hash-locking for fungible tokens only~(F3, F4 absent). Xie~\etal~\cite{zkbridge2022} build a trust-minimised zero-knowledge Cosmos--Ethereum bridge with generic, semantics-agnostic messaging. Mari\'{c}~\etal~\cite{maric2025} formally verify a bridge implementation in Isabelle/HOL; we machine-check our protocol's state machine with TLC (Section~\ref{sec:eval-security}) but leave implementation-level verification open (Section~\ref{sec:limitations}).

\subsection{Protocol Specifications and Standards}

The IETF SATP~\cite{satp2023} targets atomic burn-and-mint transfers via trusted gateways executing 2PC---conceptually adjacent, but gateway-trusted rather than cryptographically locked, specification-stage rather than deployed, and without semantic preservation~(F3, F4). Cosmos IBC~\cite{ibc} natively connects Cosmos~SDK chains: ICS-20~\cite{ics20} moves single-denomination fungible tokens and ICS-721~\cite{ics721} unique NFTs---neither the semi-fungible multi-class semantics indicators require, and neither with HTLC atomicity~(F1) or cryptographic class-identity binding~(F3). Our protocol uses IBC indirectly through Axelar and adds both above it.

\subsection{Production Bridge Systems}

Production bridges trade cryptographic mechanisms for chain coverage under validator-quorum trust: \textit{Gravity Bridge}~\cite{gravity2022}, \textit{Wormhole}~\cite{wormhole2024}, \textit{LayerZero}~\cite{layerzero2021}, \textit{Chainlink CCIP}~\cite{ccip2023}, and \textit{Axelar ITS}~\cite{axelar2024its} support fungible (typically single-class) tokens and provide none of F1--F5 as protocol primitives; optimistic designs such as Hop~\cite{hop2021} substitute fraud proofs, again without F1 or F3--F4. The exploit record---most prominently the \$190M Nomad replay hack~\cite{nomad2022}---and the broader threat landscape~\cite{zhang2024bridges} motivate cryptographic locks over pure honest-majority designs.

Table~\ref{tab:bridges} compares these approaches; each ``No'' denotes absence as a \textit{built-in protocol-layer primitive} (a generic messaging bridge can compose F1--F4 at the application layer, but none surveyed offers them natively). Two contributions follow: (i)~no surveyed bridge integrates HTLC atomicity with an automatic callback closing the heterogeneous-chain timeout race; (ii)~no surveyed approach, at any layer, preserves domain-specific token semantics across transfer~\cite{zhang2024bridges,li2025bridges,belchior2021}.

\subsection{Tokenised Environmental Assets}

On the domain side, Toucan~\cite{toucan2022} tokenises registry-verified carbon credits onto Polygon as pooled assets (BCT, NCT), and KlimaDAO~\cite{klimadao2024,frontiers2024klimadao} builds carbon DeFi on top. Both distinguish broad credit categories but neither cryptographically binds classes to a methodology and evidence bundle~(F3) nor addresses cross-chain transfer with semantic preservation~(F4)---the problem our protocol solves.

\begin{table}[t]
\centering
\caption{Comparison of cross-chain approaches against the proposed features. ``No'' denotes \textit{absence as a protocol-layer primitive}; features F1--F2 can in principle be obtained by composing application-layer contracts on top of any general-purpose messaging bridge.}
\label{tab:bridges}
\resizebox{\columnwidth}{!}{%
\begin{tabular}{@{}lccccccc@{}}
\toprule
\textbf{Feature (protocol layer)} & \textbf{SATP} & \textbf{Gravity} & \textbf{Worm.} & \textbf{LZ} & \textbf{CCIP} & \textbf{ITS} & \textbf{Ours} \\
\midrule
F1: HTLC atomicity & No\textsuperscript{a} & No & No & No & No & No & Yes \\
F2: Auto. callback & No & No & No & No & No & No & Yes \\
F3: Determ. class identity & No & No & No & No & No & No & Yes \\
F4: Cross-chain class check & No & No & No & No & No & No & Yes \\
F5: Multi-class transfer & No & No & No & No & No & No & Yes \\
\midrule
Chain coverage & N/A\textsuperscript{b} & 2 & 30+ & 100+ & 60+ & 20+ & 2 \\
Bidirectional & Yes & Yes & Yes & Yes & Yes & Yes & Yes \\
\bottomrule
\multicolumn{8}{@{}l}{\scriptsize \textsuperscript{a}SATP uses 2PC via trusted gateways. \textsuperscript{b}SATP is specification-stage.} \\
\multicolumn{8}{@{}l}{\scriptsize Worm.=Wormhole; LZ=LayerZero; CCIP=Chainlink CCIP; ITS=Axelar ITS.}
\end{tabular}%
}
\end{table}

\section{System Design}
\label{sec:design}

This section describes the architecture (Section~\ref{sec:design-arch}), the indicator-token model that prevents semantic mixing (Section~\ref{sec:indicator-model}), the end-to-end protocol (Section~\ref{sec:flow}), and its state machine and invariants (Section~\ref{sec:state-machine}). The design is relay- and contract-environment-agnostic, requiring only that the destination support state, contract calls, and inbound/outbound relayed messages; Section~\ref{sec:implementation} describes one instantiation.

\subsection{Architecture Overview}
\label{sec:design-arch}

Figure~\ref{fig:architecture} shows the components and their interactions.

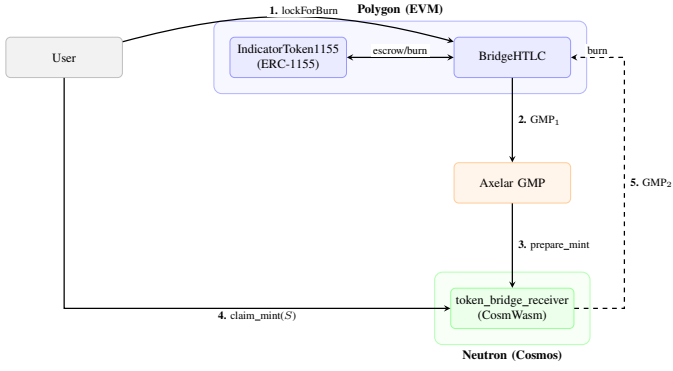
\begin{figure}[t]
\centering
\resizebox{\columnwidth}{!}{%
\begin{tikzpicture}[
    box/.style={rectangle, draw, rounded corners=3pt, minimum width=2.6cm, minimum height=0.9cm, align=center, font=\footnotesize},
    evmbox/.style={box, fill=blue!8, draw=blue!40},
    cosmosbox/.style={box, fill=green!8, draw=green!40},
    relaybox/.style={box, fill=orange!8, draw=orange!40},
    userbox/.style={box, fill=gray!10, draw=gray!50},
    arrow/.style={->, >=stealth, thick},
    dashedarrow/.style={->, >=stealth, thick, dashed},
    lbl/.style={font=\scriptsize, fill=white, inner sep=1.5pt},
]
\node[userbox] (user) at (0,0) {User};
\node[evmbox] (token) at (5,0) {IndicatorToken1155\\(ERC-1155)};
\node[evmbox] (bridge) at (10,0) {BridgeHTLC};
\node[relaybox] (axelar) at (10,-2.8) {Axelar GMP};
\node[cosmosbox] (cosmos) at (10,-5.6) {token\_bridge\_receiver\\(CosmWasm)};
\begin{scope}[on background layer]
    \node[draw=blue!30, fill=blue!3, rounded corners=5pt, fit=(token)(bridge), inner sep=10pt, label={[font=\footnotesize\bfseries]above:Polygon (EVM)}] {};
    \node[draw=green!30, fill=green!3, rounded corners=5pt, fit=(cosmos), inner sep=10pt, label={[font=\footnotesize\bfseries]below:Neutron (Cosmos)}] {};
\end{scope}
\draw[arrow, <->] (token) -- node[lbl, above] {escrow/burn} (bridge);
\draw[arrow] (user) to[bend left=15] node[lbl, above, pos=0.55] {\textbf{1.} lockForBurn} (bridge);
\draw[arrow] (bridge.south) -- node[lbl, right, xshift=2pt] {\textbf{2.} GMP$_1$} (axelar.north);
\draw[arrow] (axelar.south) -- node[lbl, right, xshift=2pt] {\textbf{3.} prepare\_mint} (cosmos.north);
\draw[arrow] (user.south) -- (0,-5.6) -- node[lbl, below] {\textbf{4.} claim\_mint($S$)} (cosmos.west);
\draw[dashedarrow] (cosmos.east) -- (12.5,-5.6) -- node[lbl, right, xshift=2pt] {\textbf{5.} GMP$_2$} (12.5, 0) -- node[lbl, above] {burn} (bridge.east);
\end{tikzpicture}}
\caption{Architecture and happy-path flow: (1)~\texttt{lockForBurn} escrows tokens; (2--3)~the relay delivers the prepare message; (4)~\texttt{claim\_mint} reveals the secret and mints; (5)~the automatic callback (dashed only to distinguish the return leg) burns the source escrow.}
\label{fig:architecture}
\end{figure}

Four components compose the system: the \textit{source multi-token contract} (class definitions and balances); the \textit{source bridge contract} (escrow, outbound prepare messages, inbound callbacks); the \textit{relay}, trusted for message \emph{authenticity} (A6) but not protocol \emph{correctness}---it cannot claim without $S$ or mint into unregistered classes, and its unavailability affects liveness only; and the \textit{destination bridge-receiver contract} (mirrored registry, sender and class validation, mint, callback emission). A per-transfer hash $H$ ties three phases together: class registration (once per class, each chain), transfer (source escrow + destination pending-mint), and claim-and-callback.

\subsection{Indicator Token Model: Preventing Semantic Mixing (F3, F5)}
\label{sec:indicator-model}

The model makes class identity content-addressed, immutable, and reproducible without coordination. Each indicator is described by an off-chain profile of its semantic attributes---indicator type (\eg, \texttt{CO2\_REMOVAL}), canonical unit (\eg, \texttt{kgCO2e}), methodology identifier (\eg, \texttt{Verra-VM0007}), and an evidence bundle---anchored on-chain by two digests, \texttt{profileHash} (canonicalised profile) and \texttt{dataHash} (evidence), which bound storage while binding the class to a verifiable artefact. The class identifier is $\texttt{indicatorId} = \text{keccak256}(\texttt{profileHash})$: under~A1, parties on different chains holding the same profile derive the same identifier with no coordination protocol.
On the source chain, each \texttt{tokenId} is immutably bound at creation to one \texttt{indicatorId} through a metadata record carrying exactly these fields; on the destination, a mirrored \texttt{TokenClass} record stores the same fields under the same \texttt{tokenId}. Fungibility holds within a class, never across classes.

\textit{Issuance and amount semantics.} Class creation is owner-only on both chains (destination registration additionally whitelist-gated); minting is restricted to the bridge contracts. Amounts are integers in the class's canonical unit, identical on both chains; cross-class unit conversion is deliberately \emph{not} performed---it would itself be semantic mixing---and unit harmonisation is out of scope (Section~\ref{sec:limitations}).

\subsection{End-to-End Flow}
\label{sec:flow}

\textit{Registration} (once per class): the owner builds the profile (Section~\ref{sec:indicator-model}), computes \texttt{profileHash}, \texttt{dataHash}, and \texttt{indicatorId}, and registers the class on both chains with the same metadata (owner-only on both; destination registration additionally whitelist-gated). Content-addressing guarantees both chains derive identical \texttt{indicatorId} values without coordination.

\textit{Transfer}: a user $U$ holding class $T$ generates a 32-byte secret $S$, computes $H = \text{keccak256}(S)$, and calls \texttt{lockForBurn}$(T, A, H, T_e, R, b)$ with amount $A$, source timeout $T_e$, destination recipient $R$, and bounty $b \geq b_{\min}$---escrowed alongside the tokens, paid to whoever finalises the source-side burn via the permissionless \texttt{claimBurn}, and returned to $U$ when the automatic callback performs the burn or on refund; the protocol-enforced minimum $b_{\min}$ makes the fallback incentive-compatible against an adversarial sender (Section~\ref{sec:disc-callback-gap}). The bridge escrows the tokens, reads \texttt{indicatorId} from the token contract (never from user input), and emits a prepare message embedding $H$, \texttt{indicatorId}, $T$, $A$, $R$, and the destination claim deadline $T_c$; only $T_c$ (not $T_e$) crosses the wire, so $\Delta T = T_e - T_c$ remains source-side. The destination validates the inbound sender and the class identity (Section~\ref{sec:state-machine}) and records a pending mint indexed by $H$.

\textit{Claim and callback}: $U$ calls \texttt{claim\_mint}$(H, S)$ before $T_c$; the contract verifies the preimage, re-validates the class, mints $A$ to $R$, and emits the $(H, S)$ callback, which the relay delivers to the source \texttt{\_execute} to burn the escrow. Since $S$ is public after the claim, \texttt{claimBurn}$(H,S)$ remains available to any party---the property the failure analysis builds on. \textit{Refund}: an unclaimed pending mint is cancelled by anyone after $T_c$; $U$ recovers escrow and bounty via \texttt{refundBurn} after $T_e$.

The callback shortens the source chain's ignorance of a destination claim from ``until $T_e$'' to one relay hop; supply conservation holds under adversarial destination timing as long as callback delay stays materially below $\Delta T$ (quantified in Section~\ref{sec:evaluation}; the prototype uses $T_e = 2$\,h, $\Delta T = 30$\,min, more than an order of magnitude above observed callback latencies).

\subsection{Protocol State Machine, Invariants, and Semantic Preservation}
\label{sec:state-machine}

Per transfer $H$, the source lock takes states $\Sigma_s = \{\bot, \textsc{Locked}, \textsc{Claimed}, \textsc{Refunded}\}$ and the destination $\Sigma_d = \{\bot, \textsc{Pending}, \textsc{Minted}, \textsc{Refunded}_d\}$, with guarded transitions: \texttt{lockForBurn} ($\bot \to \textsc{Locked}$); \texttt{prepare\_mint} ($\bot \to \textsc{Pending}$; guard: A6-authentic message referencing a registered class with matching \texttt{indicator\_id}); \texttt{claim\_mint} ($\textsc{Pending} \to \textsc{Minted}$; guard: $\text{keccak256}(S)=H \wedge t < T_c$); \texttt{\_execute}/\texttt{claimBurn} ($\textsc{Locked} \to \textsc{Claimed}$; guard: $\text{keccak256}(S)=H$); \texttt{refund\_mint} ($\textsc{Pending} \to \textsc{Refunded}_d$; guard: $t \geq T_c$); and \texttt{refundBurn} ($\textsc{Locked} \to \textsc{Refunded}$; guard: $t \geq T_e$, sender-only). \textsc{Claimed}, \textsc{Refunded}, \textsc{Minted}, and $\textsc{Refunded}_d$ are terminal; duplicate or replayed messages target non-$\bot$ states and are rejected without state change. Three invariants over $\Sigma_s \times \Sigma_d$ capture the paper's safety claims: \textbf{I1} (single settlement): each side reaches at most one terminal state; \textbf{I2} (supply conservation): $\neg(\textsc{Refunded} \wedge \textsc{Minted})$; \textbf{I3} (semantic preservation): $\textsc{Minted}$ into class $c$ implies a prior $\textsc{Locked}$ of equal amount whose source class carries the same \texttt{indicatorId} as $c$. Assumption dependencies follow Table~\ref{tab:assumptions}: I1 holds under A1--A2; I3 is the model-level projection of Theorem~\ref{thm:semantic} (within the model, whose relay may delay, drop, or duplicate but not forge messages, it is structural; the on-chain guarantee requires A1, A5, and---for amount backing---A6); I2, and therefore per-transfer atomicity as the conjunction I1$\,\wedge\,$I2, requires A4 in the callback-failure regime (Section~\ref{sec:disc-callback-gap}). Section~\ref{sec:eval-security} reports exhaustive model checking of I1--I3.

\begin{theorem}[Semantic preservation]
\label{thm:semantic}
Under A1, A5, and A6, every destination mint of amount $A$ into class $c$ is preceded by a source escrow of amount $A$ of a class $c'$ with $\texttt{indicatorId}(c') = \texttt{indicatorId}(c)$, and $c$ and $c'$ commit to the same profile.
\end{theorem}
\begin{IEEEproof}[Proof sketch]
Mints occur only in \texttt{claim\_mint}, which requires $\textsc{state}[H]=\textsc{Pending}$; \textsc{Pending} is set only by \texttt{prepare\_mint}, whose guard requires a registered destination class matching the message's \texttt{indicator\_id} and \texttt{token\_id}. By A6, the message was emitted by \texttt{lockForBurn}, which reads \texttt{indicatorId} on-chain from the source token contract rather than accepting user input, and escrows exactly the amount embedded in the message. Registration binds $\texttt{indicatorId} = \text{keccak256}(\texttt{profileHash})$ identically on both chains, so by A1 (collision resistance) equal identifiers imply equal profile hashes, and by A5 the anchored profile is the published one. Hence the minted class commits to the same profile as the escrowed class, and amounts coincide.
\end{IEEEproof}

Without A6, the amount clause fails (a forged prepare message can mint unbacked tokens into a registered class), but the class-identity clause survives: no execution path mints into an unregistered or identity-mismatched class.

\section{Prototype Implementation}
\label{sec:implementation}

We built a prototype that instantiates the design on Polygon (EVM source) and Neutron (Cosmos destination), with Axelar GMP as the relay. The EVM--Cosmos pair is deliberate: its architectural heterogeneity (Section~\ref{sec:bg-taxonomy}) exercises every cross-chain assumption. The design itself is not CosmWasm-specific---the destination component requires only inbound message verification, hashlock-indexed mutable state, and outbound message emission---so a Cosmos~SDK chain with a native module or another smart-contract environment with these capabilities would change only the implementation.

\textit{Source-chain contracts.} \texttt{IndicatorToken1155} extends ERC-1155 with the class metadata of Section~\ref{sec:indicator-model}, owner-gated class creation, and bridge-only mint/burn. \texttt{BridgeHTLC} enforces the minimum timelock and minimum bounty, stores locks under $H$, exposes \texttt{claimBurn}/\texttt{refundBurn} plus a pull-payment \texttt{withdrawPending} for bounty payouts that fail to transfer, and implements the Axelar Executable interface so \texttt{\_execute} is the only inbound-callback path; a unit-test suite covers the bounty paths. Both compile with Solidity 0.8.19 (IR pipeline, optimiser at 1000 runs).

\textit{Destination-chain contract.} \texttt{token\_bridge\_receiver} (CosmWasm) holds the mirrored class registry, sender whitelist, and hashlock state. \texttt{prepare\_mint} enforces strict 32-byte hex formats and the semantic-identity \emph{equality} of Theorem~\ref{thm:semantic}'s guard---the message's \texttt{indicator\_id} must equal the registered class's identity, not merely reference an existing class. Inbound authorisation binds to the IBC-hook intermediary account deterministically derived from the Axelar GMP account and the canonical channel. Builds are reproducible via \texttt{cosmwasm/optimizer}.

\textit{Outbound messaging from the Cosmos side.} Prepare messages travel as JSON, callbacks as ABI bytes; the reverse direction inverts this boundary exactly. Making Cosmos-originated messages (the callback and the reverse prepare) \emph{on-chain}---rather than delegating to an off-chain relayer process---required three integration constraints that are undocumented in combination and silently stall delivery when violated (Section~\ref{sec:eval-failures}): (i)~Neutron's wrapped \texttt{MsgTransfer} with non-zero acknowledgement/timeout fees for its fee-refunder module; (ii)~an explicit relayer \texttt{fee} in the Axelar GMP memo naming the fee-collector account, without which the message is routed but never executed; and (iii)~fee denomination in an \emph{Axelar-registered} asset (we use AXL)---unregistered assets (\eg, NTRN) are error-acknowledged and refunded. Claims therefore attach two coins: native fee-refunder fees plus the AXL relayer fee.

\textit{Reverse direction.} Cosmos~$\rightarrow$~EVM is implemented symmetrically: \texttt{lock\_for\_burn} escrows tokens and bounty on the CosmWasm side, reads class identity from the registry, and emits an ABI-encoded prepare via the on-chain IBC transfer; the EVM bridge validates counterpart and class identity in \texttt{\_execute}, and \texttt{claimMint} mints and emits the JSON \texttt{claim\_burn} callback. Both directions are covered by the repository's unit-test suites.

 The full source code, deployment scripts (using \texttt{ethers.js} and \texttt{CosmJS}), and the verification artefacts cited in Section~\ref{sec:evaluation} are publicly available~\cite{github_repo}.

\section{Evaluation}
\label{sec:evaluation}

We evaluate the prototype on Polygon and Neutron mainnets, addressing four questions: (i)~atomicity and semantic consistency in practice; (ii)~end-to-end latency and its acceptability for environmental-asset workloads; (iii)~per-side on-chain cost and its decomposition into messaging vs.\ semantic-integrity layer; (iv)~safety and liveness under plausible failure modes. Transaction hashes for every run are published in the repository~\cite{github_repo} for independent verification on Polygonscan and Mintscan.

\textit{Experimental setup.} Mainnet Axelar GMP routes were used with default finality: $\sim$64~Polygon-block confirmations before relaying, instant BFT finality on Neutron. The source-paid relay fee is denominated in the source-chain native gas token (POL after the September~2024 MATIC/POL rebrand) and covers destination execution gas and validator attestation. Each transfer used $T_e = 2$~h, $\Delta T = 30$~min.

\subsection{Atomic Transfer Results}
\label{sec:eval-htlc}

\textit{Fully automatic bidirectional validation.} For each transfer we verified escrow-before-any-destination-operation, mint-only-after-reveal, burn-after-callback, and \texttt{indicatorId} consistency between the prepare message and the registered class. EVM$\rightarrow$Cosmos completed end-to-end in 108~s (34~s prepare relay, $\sim$1~s claim, 62~s \emph{on-chain automatic callback}---the destination's IBC message travelling Neutron$\rightarrow$Axelar$\rightarrow$Polygon and executing the source burn with no party's intervention); Cosmos$\rightarrow$EVM completed in 93~s (55~s prepare, $\sim$1~s claim, 31~s callback). The leg asymmetry is structural (Section~\ref{sec:bg-taxonomy}): Axelar-relayed legs are faster and less variable than legs riding public IBC relayers, and the EVM-side prepare additionally pays Axelar's stricter source-finality policy.

\textit{Latency characterization.} We ran a campaign of 50 Polygon$\rightarrow$Neutron transfers across 7 distinct days and 9 hours of the day. All 50 succeeded and \emph{all} were finalised by the on-chain automatic callback---no run required the manual fallback. End-to-end latency was 97~s median (p90 122~s, p95 141~s, max 164~s), split into a 33~s median forward relay and a 47~s median callback leg (Figure~\ref{fig:cdf}); the tail is modest (p95 is $1.45\times$ the median) and stays far below the 30-min $\Delta T$. Gas was essentially deterministic (\texttt{lockForBurn} 639{,}191~gas in 49 of 50 runs). These latencies are orders of magnitude faster than voluntary-carbon-market clearing today~\cite{frontiers2024klimadao}.

\begin{figure}[t]
\centering
\includegraphics[width=0.72\columnwidth]{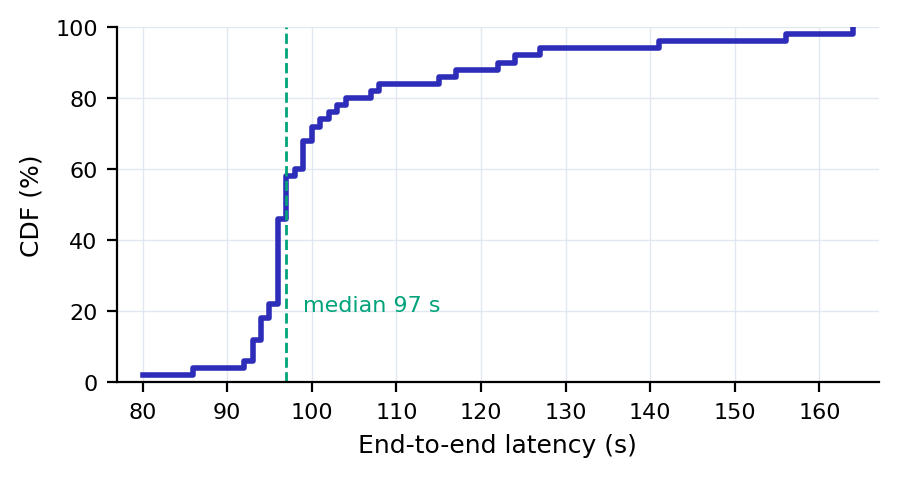}
\caption{Empirical CDF of end-to-end latency over the 50-transfer Polygon$\rightarrow$Neutron campaign (7 days). Median 97~s; p95 141~s; all runs finalised by the automatic callback.}
\label{fig:cdf}
\end{figure}

\subsection{Transaction Costs}
\label{sec:eval-costs}

The per-transfer costs in the EVM$\rightarrow$Cosmos direction are: on Polygon, \texttt{lockForBurn} (615{,}645~gas, including bounty escrow and GMP emission) and the callback-executed burn (executed by Axelar's relayer, paid from the relay fee; the manual \texttt{claimBurn} fallback costs 113{,}406~gas including the bounty payout); on Neutron, \texttt{claim\_mint} (452{,}000--467{,}000~gas, including the on-chain IBC callback emission). In the reverse direction, \texttt{lock\_for\_burn} costs 457{,}585~gas on Neutron and \texttt{claimMint} 225{,}752~gas on Polygon. Per transfer, the user pays the source-chain gas, the forward relay fee ($\sim$1~POL), the destination fee-refunder fees (0.4~NTRN, half refunded on success), the callback relay fee ($\sim$1~AXL, unused portion refundable by Axelar's gas service), and the optional bounty (returned unless the fallback path is exercised). The cost is dominated by the two relay fees rather than on-chain gas; at representative mainnet prices during the campaign the total user-paid cost is on the order of US\$1--2 per transfer. The total remains orders of magnitude below the value of typical indicator tokens (dollars per kgCO$_2$e or per kWh), supporting the viability claim in Section~\ref{sec:disc-implications}. The refund path was validated by executing deliberately unclaimed transfers; refund operations recover the escrowed tokens \emph{and} the bounty, so the timeout path imposes no asymmetric penalty on users.

\subsection{Cost Decomposition: Where Does the Overhead Go?}
\label{sec:eval-overhead}

Decomposing \texttt{lockForBurn} (639{,}191~gas), the Axelar GMP gas service dominates ($\sim$56\%, $\sim$345{,}000~gas), followed by HTLC state storage and JSON payload construction. The semantic-integrity layer---a single \texttt{SLOAD} of \texttt{indicatorId} plus 64~bytes of payload---costs $\sim$2{,}100~gas, under 0.4\% of the total: against an otherwise-identical HTLC bridge without semantic validation, preserving semantics adds essentially nothing. The cross-chain machinery, not the semantic layer, is what costs.

\subsection{Failure-Mode Analysis: Safety and Liveness}
\label{sec:eval-failures}

A bridge must satisfy two informal properties: \textit{safety}---bad outcomes (double-spend, semantic mixing, theft) do not occur---and \textit{liveness}---good outcomes (eventual claim or refund) eventually occur. The protocol preserves both unconditionally in four of five failure modes; the fifth weakens safety to a conditional property. (a)~\textit{Forward relay failure (GMP$_1$ lost):} no destination state is created; the user refunds via \texttt{refundBurn} after $T_e$. (b)~\textit{Callback relay failure with monitor present:} $S$ becomes publicly readable on the destination once \texttt{claim\_mint} succeeds, and any observer can submit \texttt{claimBurn}$(H, S)$ on the source before $T_e$; the corrective window of $\Delta T = 30$~min is large relative to typical Axelar latencies ($\leq$1~min). (c)~\textit{Callback relay failure with no monitor before $T_e$:} the only mode in which safety is not unconditional. The user can call \texttt{refundBurn} after $T_e$, leaving destination tokens in existence---supply inflation. The capable set is unbounded (any destination observer), and the escrowed bounty makes the motivated set nonempty: any rational observer profits from submitting \texttt{claimBurn}. The residual assumption is one rational, live observer acting before $T_e$ (A4); we name the regime in which no such observer exists \textit{weak supply conservation} (Sections~\ref{sec:disc-callback-gap},~\ref{sec:limitations}). (d)~\textit{Invalid secret:} \texttt{claim\_mint} and \texttt{claimBurn} require exactly 32 bytes, \texttt{0x}-prefixed hex, hashing to $H$; the fixed-length check rules out secret-encoding malleability, a known source of real HTLC bugs. Invalid secrets are rejected without state change. (e)~\textit{Unregistered class on destination:} \texttt{prepare\_mint} rejects messages whose \texttt{token\_id} has no registered class with matching identity, no destination state is created, and the user refunds after $T_e$; semantic mixing is prevented by construction.

\textit{Deployment fault observations.} The deployment exercised several failure paths, with verifiable transaction evidence~\cite{github_repo}: an \emph{undelivered prepare} left no destination state and was recovered via \texttt{refundBurn}; a \texttt{claim\_mint} violating the destination chain's IBC fee rules \emph{aborted atomically} (no mint, pending state still claimable); and callbacks emitted without a valid relayer fee were \emph{routed but never executed} (modes~b/c)---in one such episode the escrow stayed locked until a monitor collected the bounty via \texttt{claimBurn} 1{,}258~s after lock, a live demonstration of the incentive-compatible fallback under a genuinely failed callback. Safety held in every episode and every locked asset was recovered through the designed paths.

\subsection{Security Argument}
\label{sec:eval-security}

We argue that the protocol satisfies three properties: \textbf{P1 (atomicity)} a transfer either completes on both chains or leaves no net token movement; \textbf{P2 (supply conservation)} the global token supply across both chains never increases; \textbf{P3 (semantic preservation)} no token of class $X$ on the source ever produces a token of class $Y\neq X$ on the destination. P2 and P3 correspond to invariants I2 and I3 of Section~\ref{sec:state-machine}; P3 is established by Theorem~\ref{thm:semantic}. The arguments below sketch how each property is established and which assumptions (Section~\ref{sec:threat-model}) it depends on.

Dependencies follow Table~\ref{tab:assumptions}. P1 (atomicity) rests on A1+A2---A1 makes the hashlock binding (no $S'\neq S$ with $\text{keccak256}(S')=H$), A2 prevents either chain from rewinding a confirmed state transition---and, like P2, is \emph{conditional on A4 in the callback-failure regime}: in mode~(c) a destination mint coexists with a source refund, violating P1 and P2 simultaneously unless a monitor acts. P2 additionally rests on A6 (a relay that can forge prepare messages can mint unbacked tokens into registered classes) and on A3 (the destination eventually learns of the lock so a claim is possible). P3's class-identity guarantee rests on A1+A5: A1 makes \texttt{indicatorId} a binding commitment to the profile, and A5 ensures the profile published by the issuer at registration is the one the on-chain hash actually commits to; the on-chain identity check in \texttt{prepare\_mint} then forecloses any mint into an unregistered or mismatched class. The \emph{amount-backing} clause of Theorem~\ref{thm:semantic} additionally requires A6.

Modes~(a), (b), (d), and~(e) of Section~\ref{sec:eval-failures} preserve all three properties under these assumptions: in~(a) and~(e) no destination state is created and the escrow is refundable, so the properties hold vacuously; in~(b) the public secret lets any observer restore P2 via \texttt{claimBurn} before $T_e$; in~(d) A1 plus the strict 32-byte preimage check rules out any accepted $S' \neq S$, and rejected calls leave state unchanged.

 Mode~(c) is the only regime in which P1 and P2 are conditional. Under A1+A2, $S$ is still public on the destination after \texttt{claim\_mint}, so the capable set for invoking \texttt{claimBurn} is unbounded; both properties hold iff at least one party in that set acts before $T_e$ (assumption A4). Concretely, the dependency reduces to: ``with probability~1, some party that observes the destination's \texttt{Claimed} event submits the source-side \texttt{claimBurn} within $T_e - t_{\text{claim}} \geq \Delta T = 30$~min''. We refer to this conditional regime as weak supply conservation (Section~\ref{sec:disc-callback-gap}) and discuss the three mitigations that close or reduce the conditional in deployment.

\textit{Model checking.} We formalised the state machine of Section~\ref{sec:state-machine} in TLA\textsuperscript{+}~\cite{lamport2002} and exhaustively model-checked it with TLC over a discretised clock ($T_c{=}4$, $T_e{=}6$; a second configuration uses $T_c{=}3$, $T_e{=}5$) and a relay that may delay, drop, or duplicate either message (duplicate deliveries are verified to be state no-ops). TLC mechanically confirms the analysis above: without A4, I2 is violated---the checker returns an 11-state counterexample in which the callback is merely \emph{delayed} past $T_e$ and \texttt{refundBurn} then coexists with a destination mint---while I1 and I3 hold over the exhaustive state space (157 distinct states). With A4, modelled as an oracle guard (``if the secret is public, a monitor claims before the refund executes''), all invariants hold exhaustively (145 distinct states). The mechanised claim is therefore \emph{I2 holds under A4-as-assumption}, matching Section~\ref{sec:disc-callback-gap}; the monitor's liveness itself is not verified. Two scope notes bound what the model establishes: I2 is the \emph{per-transfer} projection of P2---violations of P2 by forged mints (A6) lie outside the model, whose relay is faulty but not Byzantine---and the model covers a single transfer instance. The specification, configurations, and reproducible logs are published in the repository~\cite{github_repo}.

\section{Discussion}
\label{sec:discussion}

\subsection{Critical Analysis of the Results}
\label{sec:disc-critical}

The principal positive finding is that domain-specific integrity is essentially free once messaging is paid for; the principal unanticipated finding is that the dominant operational risk on the callback path is not relay \emph{latency} but relay \emph{eligibility}---a message whose fee is missing or wrongly denominated is routed and then silently never executed, a failure only the monitor path (or the timeout) bounds---though across the 50-transfer campaign every run completed automatically, indicating that under normal relayer conditions the eligibility risk is a deployment-configuration concern rather than a recurring runtime failure. Two scope caveats remain: \textit{atomicity is not unconditional} (the callback-failure regime of Section~\ref{sec:disc-callback-gap}, which we resist rounding to ``atomic in practice''), and \textit{semantic mixing is prevented at the bridge level only}---downstream pooling after minting reproduces the BCT failure mode, out of scope here.

\subsection{Weak Supply Conservation: A Conditional-Safety Property}
\label{sec:disc-callback-gap}

We name the safety property under callback failure \textit{weak supply conservation}: strictly weaker than Herlihy's unconditional supply conservation~\cite{herlihy2018} in that it additionally requires A4. The weakening is structural: in a single-owner \textit{transfer} (unlike a two-party \textit{exchange}) the user already holds the destination tokens when the source burn is needed, so without an incentive \texttt{claimBurn} would be rewardless.

The implemented mitigation is the \textit{bountied callback}: the escrowed bounty is paid to whoever submits a valid \texttt{claimBurn} (a pull-payment fallback ensures payout failure cannot block settlement), turning the empty motivated set into ``any rational observer''---the assumption class watchtower systems (\eg, Lightning) deploy in production, and the one TLC verifies as sufficient (Section~\ref{sec:eval-security}).

\textit{Adversarial-sender analysis.} The party that profits in mode~(c) is the \emph{sender}, who also funds the callback fee as claimer and can thus deliberately induce callback failure (mis-denominating that fee, a class we observed operationally). With an \emph{optional} bounty this defeats incentive compatibility for free (lock with $b=0$, claim, sabotage the callback, refund). The protocol therefore enforces $b \geq b_{\min}$ on both chains, $b_{\min}$ owner-set above the monitor's execution cost (default 0.05~POL). Since the secret is public after the claim and \texttt{claimBurn} is permissionless, a sabotaged callback funds its own defeat: the attacker's best case is paying $b_{\min}$ to a monitor for no gain. Two stronger mitigations remain future work: accountable \textit{watchtowers}~\cite{herlihy2018} and \textit{native IBC callbacks}, which bypass the relay and collapse A4 into A3.

\subsection{Comparison with Existing Approaches}
\label{sec:disc-comparison}

 We adapt HTLCs from atomic exchange to atomic transfer and replace the timeout-ordering safety argument with a relay-driven callback. When correctly funded, the on-chain callback leg has a 47~s median and a 79~s p95 over the 50-transfer campaign (31~s in the reverse direction)---more than an order of magnitude below the 30-min $\Delta T$, mitigating the timeout race~\cite{thyagarajan2022} without the cryptographic apparatus of swap-focused work. The price is the relay dependency on the callback hop---and, as the deployment observations show, that dependency includes fee-eligibility conditions under which a callback stalls \emph{indefinitely} rather than slowly, making the bountied monitor path load-bearing rather than auxiliary (Section~\ref{sec:disc-callback-gap}).

Against alternatives: SATP~\cite{satp2023} reaches the same end state via trusted 2PC gateways---opposite ends of the trust/coordination trade-off. Validator-quorum bridges (Wormhole, LayerZero, CCIP, ITS) offer far broader coverage; we share their honest-majority assumption for message \emph{authenticity} (A6), but a Byzantine quorum here is confined to inflation within registered classes---it cannot cross class boundaries (Theorem~\ref{thm:semantic}) nor release escrow without the published secret. None preserves domain-specific semantics, so the BCT failure mode could re-emerge if any were used for indicator tokens unmodified.

\subsection{Implications}
\label{sec:disc-implications}

 A 97~s median latency, a per-transfer cost on the order of US\$1--2, and $<$0.4\% semantic overhead make the approach viable for environmental-asset workloads (retirement, compliance reporting, periodic transfers between venues), though not for high-frequency movement. The indicator-identity construction is domain-agnostic: any multi-class asset (fractionalised real-world assets, credentials) faces the same mixing risk and admits the same solution at essentially zero added gas.

\subsection{Limitations}
\label{sec:limitations}

The user must reveal $S$ before $T_c$ (inherent to HTLC atomicity), and the $\geq$30-min monitoring window is a real dependency---the bounty makes acting within it rational but still requires a live observer. Only hashes are anchored, so a malicious distributor can equivocate on the off-chain profile; a content-addressed registry (IPFS CID anchoring) would close this vector. Class creation is owner-gated. The latency campaign covers EVM$\rightarrow$Cosmos ($n=50$); the reverse direction is validated end-to-end but characterised by a smaller sample ($n=2$). The prototype is Polygon--Neutron over Axelar only---portability to other pairs is argued as a capability checklist, not demonstrated---and the state machine is machine-checked at the design level only; implementation-level verification~\cite{maric2025} remains open.

\section{Conclusions and Future Work}
\label{sec:conclusion}

We addressed atomicity across heterogeneous chains and semantic mixing during transfer, with two contributions: the \textit{automatic callback with an escrowed minimum bounty} (HTLCs adapted from exchange to transfer, the failure regime covered by an incentive-compatible monitor path machine-checked as sufficient for supply conservation), and the content-addressed \textit{deterministic indicator identity} that prevents mixing by construction (Theorem~\ref{thm:semantic}), addressed by no existing approach. Mainnet evaluation showed fully automatic bidirectional transfers over 50 runs across 7 days (97~s median) and documented the integration failure modes of Cosmos-originated messaging. \textit{Future work}: native IBC callbacks (eliminating the monitor assumption), DAO-governed class registration, and implementation-level formal verification~\cite{maric2025}.

\section*{Acknowledgment}

This work was supported by ANII (Uruguay) under grant POS\_NAC\_2023\_4\_178540 and by Pyxis. The authors used Claude (Anthropic) in agent mode as a coding assistant for implementation; all AI-assisted code was reviewed, tested, and validated by the authors.

\balance
\bibliographystyle{IEEEtran}

\begin{thebibliography}{32}

\bibitem{belchior2021}
R.~Belchior, A.~Vasconcelos, S.~Guerreiro, and M.~Correia, ``A Survey on Blockchain Interoperability: Past, Present, and Future Trends,'' \emph{ACM Computing Surveys}, vol.~54, no.~8, pp.~1--41, 2021.

\bibitem{wang2021}
G.~Wang, Q.~Wang, and S.~Chen, ``Exploring Blockchains Interoperability: A Systematic Survey,'' \emph{ACM Computing Surveys}, vol.~55, no.~13s, Art. no. 290, Jul. 2023.

\bibitem{zhang2024bridges}
L.~Zhang \emph{et al.}, ``Security of Cross-Chain Bridges: Attack Surfaces, Defenses, and Open Problems,'' in \emph{Proc. 27th Int. Symp. Research in Attacks, Intrusions and Defenses (RAID)}, 2024.

\bibitem{chainalysisBridges2022}
Chainalysis Team, ``Cross-Chain Bridge Hacks Emerge as Top Security Risk,'' Aug. 2, 2022. \url{https://www.chainalysis.com/blog/cross-chain-bridge-hacks-2022/}

\bibitem{pandolfi2025rapid}
S.~Pandolfi, G.~Llamb\'{i}as, and L.~Gonz\'{a}lez, ``Blockchain Interoperability for Environmental Indicators: a Rapid Review,'' in \emph{Proc. 7th Int. Conf. Blockchain Computing and Applications (BCCA)}, Dubrovnik, Croatia, 2025, pp.~857--866, doi: \href{https://doi.org/10.1109/BCCA66705.2025.11229811}{10.1109/BCCA66705.2025.11229811}.

\bibitem{frontiers2024klimadao}
A.~Ballesteros-Rodr\'{i}guez, J.~De-Lucio, and M.-\'{A}.~Sicilia, ``Tokenized carbon credits in voluntary carbon markets: the case of KlimaDAO,'' \emph{Frontiers in Blockchain}, vol.~7, Art.\ no.\ 1474540, Oct. 2024, doi: \href{https://doi.org/10.3389/fbloc.2024.1474540}{10.3389/fbloc.2024.1474540}.


\bibitem{gray1978}
J.~Gray, ``Notes on Data Base Operating Systems,'' in \emph{Operating Systems: An Advanced Course}, R.~Bayer, R.~M.~Graham, and G.~Seegm\"{u}ller, Eds. Berlin, Heidelberg: Springer, 1978, pp.~393--481.

\bibitem{herlihy2018}
M.~Herlihy, ``Atomic Cross-Chain Swaps,'' in \emph{Proc. 2018 ACM PODC}, 2018, pp.~245--254.

\bibitem{thyagarajan2022}
S.~A.~K.~Thyagarajan, G.~Malavolta, and P.~Moreno-Sanchez, ``Universal Atomic Swaps: Secure Exchange of Coins Across All Blockchains,'' in \emph{Proc. 2022 IEEE Symp. Security and Privacy (S\&P)}, 2022, pp.~1299--1316.

\bibitem{carbonplan2022}
G.~Badgley and D.~Cullenward, ``Zombies on the Blockchain,'' CarbonPlan, Apr. 2022. \url{https://carbonplan.org/research/toucan-crypto-offsets}

\bibitem{zamyatin2021}
A.~Zamyatin \emph{et al.}, ``SoK: Communication Across Distributed Ledgers,'' in \emph{Proc. Financial Cryptography and Data Security (FC)}, 2021, pp.~3--36.

\bibitem{satp2023}
M.~Hargreaves, T.~Hardjono, and R.~Belchior, ``Secure Asset Transfer Protocol (SATP) Core,'' Internet Engineering Task Force, Internet-Draft, draft-ietf-satp-core, 2024. \url{https://datatracker.ietf.org/doc/draft-ietf-satp-core/}

\bibitem{wormhole2024}
Wormhole Foundation, ``Wormhole: Generic Message Passing Protocol---Architecture and Security Overview,'' \url{https://docs.wormhole.com/wormhole/}

\bibitem{layerzero2021}
R.~Zarick, B.~Pellegrino, and C.~Bhageria, ``LayerZero: Trustless Omnichain Interoperability Protocol,'' LayerZero Labs, Whitepaper V1.1.0, 2021. \url{https://layerzero.network/publications/LayerZero_Whitepaper_V1.1.0.pdf}

\bibitem{ccip2023}
Chainlink Labs, ``Chainlink 2.0: Next Steps in the Evolution of Decentralized Oracle Networks,'' Whitepaper, Apr. 2021. \url{https://research.chain.link/whitepaper-v2.pdf}

\bibitem{axelar2024its}
Axelar Network, ``Interchain Token Service Documentation,'' \url{https://docs.axelar.dev/dev/send-tokens/interchain-tokens/intro/}

\bibitem{zkbridge2022}
T.~Xie \emph{et al.}, ``zkBridge: Trustless Cross-chain Bridges Made Practical,'' in \emph{Proc. 2022 ACM SIGSAC Conf. Computer and Communications Security (CCS)}, 2022, pp.~3003--3017.

\bibitem{robinson2022}
P.~Robinson, ``Survey of Crosschain Communications Protocols,'' \emph{Computer Networks}, vol.~200, Art.\ no.\ 108488, Dec. 2021.

\bibitem{axelar}
S.~Gorbunov and G.~Vlachos, ``Axelar Network: Connecting Applications with Blockchain Ecosystems,'' Axelar, Whitepaper, Jan. 2021. \url{https://docs.axelar.dev/axelar_whitepaper.pdf}

\bibitem{ibc}
C.~Goes, ``The Interblockchain Communication Protocol: An Overview,'' \emph{arXiv preprint arXiv:2006.15918}, Jun. 2020.

\bibitem{li2025bridges}
N.~Li \emph{et al.}, ``Blockchain Cross-Chain Bridge Security: Challenges, Solutions, and Future Outlook,'' \emph{Distrib. Ledger Technol.: Res. Pract.}, vol.~4, no.~1, Art.\ no.\ 8, Feb. 2025.

\bibitem{erc1155}
W.~Radomski, A.~Cooke, P.~Castonguay, J.~Therien, E.~Binet, and R.~Sandford, ``EIP-1155: Multi Token Standard,'' Ethereum Improvement Proposals, 2018. \url{https://eips.ethereum.org/EIPS/eip-1155}

\bibitem{guo2024}
H.~Guo \emph{et al.}, ``A Framework for Efficient Cross-Chain Token Transfers in Blockchain Networks,'' \emph{J. King Saud Univ.---Comput. Inf. Sci.}, vol.~36, no.~2, Art.\ no.\ 101968, Feb. 2024.

\bibitem{maric2025}
O.~Mari\'{c}, B.~Scholz, and P.~Suboti\'{c}, ``Formal Verification of a Fail-Safe Cross-Chain Bridge,'' in \emph{Proc. 6th Int. Workshop on Formal Methods for Blockchains (FMBC)}, 2025.

\bibitem{lamport2002}
L.~Lamport, \emph{Specifying Systems: The TLA\textsuperscript{+} Language and Tools for Hardware and Software Engineers}. Boston, MA: Addison-Wesley, 2002.

\bibitem{ics20}
Interchain Foundation, ``ICS 20: Fungible Token Transfer,'' Cosmos IBC Specification, 2019. \url{https://github.com/cosmos/ibc/blob/main/spec/app/ics-020-fungible-token-transfer/README.md}

\bibitem{ics721}
Interchain Foundation, ``ICS 721: Non-Fungible Token Transfer,'' Cosmos IBC Specification, 2022. \url{https://github.com/cosmos/ibc/blob/main/spec/app/ics-721-nft-transfer/README.md}

\bibitem{gravity2022}
Gravity Bridge, ``Gravity Bridge Design Overview,'' GitHub, 2022. \url{https://github.com/cosmos/gravity-bridge/blob/main/docs/design/overview.md}

\bibitem{hop2021}
C.~Whinfrey, ``Hop: Send Tokens Across Rollups,'' Hop Protocol, Whitepaper, 2021. \url{https://hop.exchange/whitepaper.pdf}

\bibitem{nomad2022}
Immunefi, ``Hack Analysis: Nomad Bridge, August 2022,'' Aug. 2022. \url{https://medium.com/immunefi/hack-analysis-nomad-bridge-august-2022-5aa63d53814a}

\bibitem{toucan2022}
Toucan Protocol, ``Tokenization of Carbon Credits,'' 2022. \url{https://toucan.earth/}

\bibitem{klimadao2024}
KlimaDAO, ``Digital Carbon Ecosystem,'' \url{https://docs.klimadao.finance/ecosystem/digital-carbon}

\bibitem{github_repo}
S.~Pandolfi, ``evm-cosmos-indicator-htlc: Atomic Cross-Chain Bridge for Environmental Indicator Tokens,'' GitHub Repository, 2026. \url{https://github.com/sebapandolfi/evm-cosmos-indicator-htlc}





\end{thebibliography}

\end{document}